\documentclass{llncs}
\usepackage{euscript}
\usepackage{amssymb}
\usepackage[all]{xy}
\usepackage{url}

\newcommand{\R}{\mathcal{R}}
\newcommand{\Partes}{\mathcal{P}}

\newcommand{\M}{\mathcal{M} }
\newcommand{\D}{\mathcal{D} }

\newcommand{\quo}[2]{{#1}/{#2}}
\newcommand{\<}{\sqsubseteq}

\newcommand{\disc}{\mathcal{D}}
\newcommand{\lra}{\longrightarrow}
\newcommand{\mult}{\mathcal{M}}
\newcommand{\N}{\textup{I}\!\textup{N}}
\newcommand{\ps}{\mathcal{P}}
\newcommand{\Rel}{\textup{\textbf{Rel}}}
\newcommand{\rel}[1]{\mathrm{Rel}(#1)}

\newcommand{\Sets}{\textup{\textbf{Sets}}}

\begin{document}

\title{Multiset bisimulations as a common framework for ordinary and 
probabilistic bisimulations\thanks{Research supported by the Spanish
projects DESAFIOS TIN2006-15660-C02-01 and PROMESAS S-0505/TIC/0407}}

\author{David de Frutos Escrig 
\and Miguel Palomino \and Ignacio F\'abregas}
\institute{Departamento de Sistemas Inform\'aticos y Computaci\'on\\ Universidad
Complutense de Madrid}

\maketitle

\begin{abstract}
Our concrete objective is to present both ordinary bisimulations
and probabilistic bisimulations in a common coalgebraic framework
based on multiset bisimulations. For that we show how to relate
the underlying powerset and probabilistic distributions functors with
the multiset functor by means of adequate natural transformations.
This leads us to the general topic that we investigate in the paper:
a natural transformation from a functor $F$ to another $G$ transforms
$F$-bisimulations into $G$-bisimulations but, in general, it is not possible
to express $G$-bisimulations in terms of $F$-bisimulations. However,
they can be characterized by considering Hughes and Jacobs' notion
of simulation, taking as the order on the functor $F$ the equivalence
induced by the epi-mono decomposition of the natural transformation
relating $F$ and $G$. We also consider the case of alternating probabilistic
systems where non-deterministic and probabilistic choices are mixed, although
only in a partial way, and extend all these results to categorical
simulations.
\end{abstract}

\section{Introduction}
Bisimulations are the adequate way to capture behavioural indistinguishability 
of states of systems. 
Ordinary bisimulations were introduced \cite{Park81} to cope with 
labelled transition systems and other similar models and have been used to 
define the formal observational semantics of many popular languages and 
formalisms, such as CCS. Bisimilarity is also the natural way to express
equivalence of states in any system described by means of a coalgebra over
an arbitrary functor $F$.
The general categorical definition can be presented in a more concrete
way for the class of polynomial functors, that are defined by means of a 
simple signature of constructors and whose properties, including
the definition of relation lifting, can be studied by means of
structural induction. 
In particular, the powerset constructor is one of them,
and therefore the class of labelled transition systems can be
studied as a simple and illustrative example of the categorical
framework.

The simplicity and richness of the theory of bisimulations made it
interesting to define several extensions in which the
structure on the set of labels of the considered systems was
taken into account, instead of the plain approach made by 
simple (strong) bisimulations. For instance, weak bisimulation takes into
account the existence of non-observable actions, while timed and
probabilistic bisimulation introduce timed or probabilistic
features. In particular, the original definition of probabilistic
bisimulation for probabilistic transition systems had to capture
the fact that we should be able to accumulate the probabilities of
several transitions arriving to equivalent (bisimilar) states in
order to simulate some transition, or conversely that we should be able 
to distribute the probability of a transition among several others 
connecting the same states.

The classical definition by Larsen and Skou \cite{LarsenSkou91} certainly 
generalizes the definition of ordinary bisimulation in a nice way, although at
the cost of leaving out the categorical scenario discussed above.
However, Vink and Rutten proved in \cite{VinkRutten99} that we can
reformulate that definition in a coalgebraic way. For that, they needed
to consider a functor $\D$ defining probabilistic distributions,
that appears as the primitive construction in the definition of the 
corresponding probabilistic systems. Even if this is quite an elegant
characterization, it forces us to leave the realm of (probabilistic)
transition systems, getting into the more abstract one of
probabilistic distributions. 


We would like to directly manage probabilistic transition
systems in order to compare the results about ordinary transition systems and 
those on probabilistic systems as much as possible. 
We have found that multi-transition systems, where we
can have several identical transitions and the number of times any of them 
appears matters, constitute the adequate framework to
establish the relation between those two kinds of transition
systems. As a matter of fact, we will see that the use of multisets
instead of just plain sets leads us to a natural presentation of
relation lifting for that construction; besides, we can add the
corresponding functor to the collection defining polynomial
functors, thus obtaining an enlarged class with similar nice properties
to those in the original class.

Although a general theory combining non-deterministic and probabilistic choices
seems quite hard to develop, since it is difficult to combine both
functors in an smooth way \cite{VarWins06}, we will present at
least the
case of \emph{alternating}\footnote{Although we call alternating to uor
systems, we
do not need an strict alternation between non-deterministic and probabilistic
states as that appearing in \cite{RGA95}, but only that these two kind of
choices would not appear mixed after the same state.} probabilistic systems.
In those systems, 
the classical definitions of ordinary and probabilistic
bisimulation can
be combined to obtain the natural definition of alternating probabilistic
bisimulation, that perfectly fits into our framework based on categorical
simulations on our multi-transition systems.

The functors defining ordinary transition systems and
probabilistic systems can be obtained by applying an adequate
{\em natural transformation} to a functor defining multiset
transition systems. In both cases bisimulations
are preserved in both directions when applying those
transformations. This leads us to the general theory that we investigate in this
paper: as is well-known, any natural transformation 
between two functors $F$ and $G$ transforms $F$-bisimulations into
$G$-bisimulations; in addition, and more interesting, 
whenever the natural transformation relating $F$ and $G$ is an epi,
we can reflect $G$-bisimulations and express them at the level of
the functor $F$, though this cannot be done in general
just by means of $F$-bisimulations. However, they can be
characterized by using Hughes and Jacobs' notion of simulation 
\cite{HughesJacobs04}, when we consider as the order on the functor $F$ 
the equivalence induced by the epi-mono decomposition of the natural 
transformation relating $F$ and $G$.
Once categorical simulations have come into play, it is nice to find that 
we can extend all our results to include  simulations based on any 
order. 
These extensions could be considered the main results in the paper, since all
our previous results on bisimulations could be presented as particular cases
of the former, using the fact that bisimulations are a particular case of
categorical simulations.

Although in a different direction, that of exploring the relation between
non-deterministic and probabilistic choices, instead of the different notions
of distributed bisimulstions, in this paper we continue the work initiated in
FORTE 2007 \cite{dFRG07}, exploring the ways the general theory on
categorical bisimulations and categorical simulations can be applied to obtain
nearly for free interesting results on applied concrete cases, that without the
support of that general theory would need different non-trivial proofs.
Therefore, our work has a mixed flavour: in one hand we develop new abstracts
results that extend the general theory, thus proving the concrete general
results we need to apply that general theory; in the other hand we apply these
results to simple but important concrete concepts, that therefore are proved to
be particular cases of the rich general theory. These are only concrete
examples that we hope to extend and generalize in the near future.

\section{Basic definitions}

We review in this section standard material on coalgebras and bisimulations, 
as can be found for example in \cite{JacobsRutten97,Rutten00,Jacobs07}. 
Besides, we introduce basic notations on multisets and the corresponding 
functor $\mult$, together with some others 
for the functor $\disc$ defining discrete probabilistic distributions.

An arbitrary endofunctor $F:\Sets \lra \Sets$ can be lifted to a functor
in the category $\Rel$ of relations $\rel{F} : \Rel \lra \Rel$.
In set-theoretic terms, for a relation $R\subset X_1 \times X_2$,
\[
\rel{F}(R) = \{ \langle u,v \rangle \in FX_1 \times FX_2 \mid 
                \exists w\in F(R).\, F(r_1)(w) = u, F(r_2)(w) = v \}
\]

It is well-known that for polynomial functors $F$, $\rel{F}$ can be 
equivalently defined by induction on the structure of $F$. 
Since we will be making extensive use of the powerset functor, we next present 
how the definition particularizes to it:
\begin{eqnarray*}
\rel{\ps G}(R) & = \{ (U,V) \mid 
   & \forall u\in U.\, \exists v\in V.\, \rel{G}(R)(u,v) \land \\
 & & \forall v\in V.\, \exists u\in U.\, \rel{G}(R)(u,v) \}
\end{eqnarray*}

Multisets will be represented as sets by considering their characteristic 
function $\chi_M : X \lra \N$; similarly, discrete probabilistic distributions
are represented by discrete measures $p_D : X \lra [0,1]$, with
$\sum_{x \in X} p_D(x) = 1$.

We will use along the paper several different ways to enumerate
the ``elements'' of a multiset. 
We define the support of a multiset $M$ as the set of elements that appear
in $M$: $\{M\}_X=\{x\in X\mid \chi_M(x)>0\}$. 
We are only interested in multisets having a finite support,
so that in the following we will assume that every multiset is finite. 
Given a finite subset $Y$ of $X$ and an
enumeration of its elements $\{y_1,\ldots y_m\}$, for each tuple
of natural weights $\langle n_1,\ldots, n_m \rangle$ we will
denote by $\sum_{y_i\in Y}n_i\cdot y_i$ the multiset $M$ given by
$\chi_M(y_i)=n_i$ and $\chi_M(y)=0$ for $y\notin Y$. 
By abuse of notation we will sometimes consider sets as a particular case 
of multisets, by taking for each finite set $Y=\{y_1,\ldots y_n\}$
the canonical associated multiset $\sum_{y_i\in Y} 1\cdot y_i$.
Finally, we also enumerate the elements of a multiset by
means of a generating function: given a finite set $I$ and
$x : I\lra X$, we denote by $\{x_i\mid i\in I\}$ the multiset
$M_I$ given by $\chi_{M_I}(y)=|\{i\in I\mid x_i=y\}|$. Note that
in this case sets are just the multisets generated by an
injective generating function.

We will denote by $\mult(X)$ the set of multisets on $X$, while $\disc(X)$
represents the set of probabilistic distributions on $X$. 
Both constructions can be naturally extended to functions, thus getting the
desired functors: for $f: X \lra Y$ we define $\mult(f): \mult(X) \lra\mult(Y)$
by $\mult(f)(\chi)(y) = \sum_{f(x) = y} \chi(x)$, and 
$\disc(f) : \disc(X) \lra \disc(Y)$ by $\disc(f)(p)(y) = \sum_{f(x) = y} p(x)$.

Although the multiset and the probabilistic distributions functors are not 
polynomial, this class can be enlarged by incorporating them 
since their liftings can be defined with the following equations:
\[
\begin{array}{rcl}
\rel{\mult G}(R) &=
&\{(M,N) \mid \exists f: I\lra GX, g : I\lra GY, \textrm{generating functions
of} \\
&& \phantom{\{(M,N) \mid}
\textrm{$M$ and $N$ s.t. $\forall i\in I.\,(f(i),g(i))\in\rel{G}(R)$} \}\,;
\end{array}
\]
\[
\begin{array}{l}
\rel{\disc G}(R) = 
\{ (d^x,d^y)\in \disc(G(X))\times \D(G(Y)) \mid 
   \forall U\subseteq G(X).\,\forall V\subseteq G(Y).\, \\
\phantom{\rel{\disc(G)}(R) = \{}
\Pi_1^{-1}(U) = \Pi_2^{-1}(V) \Rightarrow \sum_{x\in U}d^x(x)=\sum_{y\in V}d^y(y) \}\,,
\end{array}
\]
where $\Pi_1$ and $\Pi_2$ are the projections of $\rel{G}(R)$ into $GX$ and 
$GY$, respectively. 

$F$-coalgebras are just functions $\alpha : X \lra FX$. 
For instance, plain labelled transition systems arise as coalgebras for the
functor $\ps(A\times X)$.
We will also consider multitransition systems, which correspond to the
functor $\mult(A\times X)$, and probabilistic transition systems, corresponding
to $\mult_1([0,1]\times A\times X)$, where we only allow multisets in which
the sum of its elements
is 1\footnote{As a matter of fact, we could simply say that we consider a
particular case
of the functor $\M$: that in which the set of actions $A$ is the set of pairs,
whose first component is in $[0,1]$; besides, we would only consider 
multisets in which $\sum n_i p_i=1$. }
: $\sum n_i\cdot (p_i,a_i,x_i)\in \mult_1([0,1]\times
A\times X)$ iff 
$\sum n_i p_i=1$.

Then, the lifting of the functor $\mult_1([0,1]\times \cdot)$ is 
defined as a particular case of that of $\mult$ by:
\[
\begin{array}{l}
\rel{\mult_1([0,1]\times \cdot)G}(R) = \\
\qquad
\{(M,N)\in \mult_1([0,1]\times GX)\times \mult_1([0,1]\times GY)\mid \\
\qquad\phantom{\{}\exists f:I\lra [0,1]\times GX, g: I\lra [0,1]\times GY,
\textrm{generating functions of $M$}\\
\qquad\phantom{\{}\textrm{and $N$ s.t.}\
 \forall i\in I.\,\Pi_1(f(i))=\Pi_1(g(i)) \land 
                              (\Pi_2(f(i)),\Pi_2(g(i)))\in \rel{G}(R) \}\,.
\end{array}
\]

A bisimulation for coalgebras $c : X\lra FX$ and $d:Y \lra FY$ is a relation
$R\subseteq X\times Y$ which is ``closed under $c$ and $d$'':
\[
\textrm{if $(x,y) \in R$ then $(c(x), d(y)) \in \rel{F}(R)$}\,.
\]
We shall use the term $F$-bisimulation sometimes to emphasize the functor
we are working with.

Bisimulations can also be characterized by means of spans, using the general
categorical definition by Aczel and Mendler~\cite{AczelMendler89}:
\[
\xymatrix@R=5.0ex{
 {X}\ar[d]_{c}   & {R} \ar[d]_{e}\ar[l]_{r_1}\ar[r]^{r_2}    & Y\ar[d]_{d} \\
 {FX}           & {FR} \ar[l]_{Fr_1}\ar[r]^{Fr_2}           & {FY}
}
\]
$R$ is a bisimulation iff it is the carrier of some coalgebra $e$ making
the above diagram commute, where the $r_i$ are the projections of $R$ into
$X$ and $Y$.

We will also need the general concept of simulation introduced by Hughes and
Jacobs~\cite{HughesJacobs04} using order on functors.
Let $F : \Sets\lra\Sets$ be a functor. 
An order on $F$ is defined by means of a functorial collection of preorders
$\sqsubseteq_X \subseteq FX\times FX$ that must be preserved by renaming:
for every $f: X\lra Y$, if $u\sqsubseteq_X u'$ then $Ff(u) \sqsubseteq_Y Ff(u')$.

Given an order $\sqsubseteq$ on $F$, a $\sqsubseteq$-simulation for
coalgebras $c: X\lra FX$ and $d: Y\lra FY$ is a relation $R\subseteq X\times Y$
such that
\[
\textrm{if $(x,y) \in R$ then $(c(x), d(y)) \in \rel{F}_\sqsubseteq(R)$}\,,
\]
where 
$\rel{F}_{\sqsubseteq}(R)$ is 
$\sqsubseteq\circ\rel{F}(R)\circ\sqsubseteq$, 
which can be expanded to
\[
\rel{F}_{\sqsubseteq}(R)=\{(u,v)\mid\exists w\in F(\R).\; u\sqsubseteq Fr_1(w)\wedge
Fr_2(w)\sqsubseteq v\} \,.
\]

One of the cases covered inmediately by this general notion of coalgebraic
simulation is that of ordinary simulations.
However, equivalence (functorial) relations, represented by $\equiv$, are 
a particular class of orders on $F$, thus generating the corresponding 
class of $\equiv$-simulations. 
As is the case for ordinary bisimulations, $\equiv$-simulations themselves 
need not be equivalence relations, but once we imposed to the equivalence
$\equiv$ the technical condition of being stable \cite{HughesJacobs04} then the
induced notion of $\equiv$-similarity remains an equivalence itself.

\begin{proposition}
For any stable functorial equivalence relation $\equiv_{X}\subseteq
FX\times FX$, the induced notion of $\equiv_a$-similarity
relating elements of $X$ for a coalgebra $a:X\lra FX$ is an equivalence 
relation. 
In particular, for the plain equality relation $=_{X}\subseteq
FX\times FX$, $=_X$-similarity coincides with plain
$F$-bisimulation.
\end{proposition}

\section{Natural transformations and bisimulations}

Natural transformations are the natural way to relate two
functors. 
Given $F$ and $G$, two functors on $\Sets$, a natural transformation 
$\alpha:F\Rightarrow G$ is defined as a family of functions 
$\alpha_{X}:FX\rightarrow GX$ such that, for all $f:X\lra Y$, 
$Gf\circ \alpha_{X}=\alpha_{Y}\circ Ff$.
We are particularly interested in the natural transformations relating 
$\mult$ and $\ps$, and those between the functors defining probabilistic 
transition systems and probabilistic distributions. 
For the sake of conciseness we will often omit the action component $A$ 
when working with these functors; this won't affect the definitions 
nor the results.

\begin{proposition}
The support of multisets,
$\{\cdot\}_{X}:\mult(X)\lra\ps(X)$, gives rise to a natural
transformation $\{\cdot\}:\mult\Rightarrow\ps$.

Similarly, $\D_{M_{X}}:\mult_{1}([0,1]\times X)\lra\disc(X)$ given by
\[
\disc_{M}(\sum n_i \cdot (p_i, x_i))(x)=\sum_{x_i=x} n_ip_i 
\] 
induces a natural transformation
$\disc_{M}:\mult_{1}([0,1]\times\cdot)\Rightarrow\disc(\cdot)$.
\end{proposition}
\begin{proof}
Let $f : X\lra Y$. We have
$(\ps f \mathop{\circ} \{\cdot\}_X)(\sum n_i\cdot x_i) = 
\ps f(\{ x_i \}) = \{ f(x_i) \} = 
\{\cdot\}_Y(\sum n_i \cdot f(x_i)) = 
(\{\cdot\}_Y\mathop{\circ}\mult f)(\sum n_i\cdot x_i)$,
which proves that $\{\cdot\}$ is a natural transformation.

In the case of $\disc_M$:
$(\disc f \mathop{\circ} \disc_{M_X})(\sum n_i \cdot (p_i,x_i))=
\disc f(\sum n_ip_i \cdot x_i) = 
\sum_{f(x_i) = y} n_ip_i \cdot y = 
\disc_{M_Y}(\sum n_i \cdot (p_i, f(x_i))) =
(\disc_{M_Y} \mathop{\circ} \mult_1f)(\sum n_i \cdot (p_i,x_i))$,
which proves that $\disc_M$ is a natural transformation.
\qed
\end{proof}

Probabilistic transition systems were defined in \cite{LarsenSkou91} as tuples
$\Partes=(Pr,Act,Can,\mu)$, where $Pr$ is a set of processes, $Act$ the set of
actions, $Can:Pr\lra \Partes(Act)$ indicates the initial offer of each
process, and $\mu_{p,a}\in\D(Pr)$ for all $p\in Pr$, $a\in Can(p)$.
Under this definition we cannot talk about ``different probabilistic
transitions'' reaching the same process, that is, whenever we have a transition
$p\stackrel{a}{\lra_{\mu}}p'$ it ``accumulates'' all possible ways to go 
from $p$ to $p'$ executing $a$.

In our opinion this is not a purely operational way to present probabilistic
systems. For instance, if we are defining the operational semantics of a
process such as $p=\frac{1}{2}a +\frac{1}{2}a$, then we would intuitively
have two different
transitions reaching the same final state $stop$, but if we were using Larsen
and Skou's original definition, we should mix them both into a single
$p\stackrel{a}{\lra_1}stop$. Certainly, we could keep these two transitions
separated under that definition, if for some reason we decided to introduce in
the set
$Pr$ two different states $stop_1$ and $stop_2$, thus obtaining
$p\stackrel{a}{\lra_{1/2}}stop_1$ and $p\stackrel{a}{\lra_{1/2}}stop_2$. Then
we observe that whether our model captures or not the existence of
two different transitions depends on the way we define our set of processes
$Pr$.

In order to get a more natural operational representation of probabilistic
systems we define them\footnote{Although Larsen and Skou defined their systems
following the reactive aproach~\cite{RGA95}, and therefore the sum of 
their probabilities is $1$ for each action $a$, we prefer to follow in this
paper the generative aproach, so that the total addition of all the
probabilities is $1$. This is just done to simplify the notation,
since all the results in this paper are equally valid for the reactive model.}
as $\M_1([0,1]\times A\times \cdot)$-coalgebras. Once we use ``ordinary''
transitions labelled by pairs $(q,a)$ to represent the probabilistic transitions
we have no problem to distinguish two ``different'' transitions
$p\stackrel{a}{\lra_{q'}}p'$, $p\stackrel{a}{\lra_{q''}}p''$, if $p'\neq p''$.
However, in such a case it would not be adequate to treat the case
$p'=p''$ in a different way. This is why we use $\M_1$ instead of $\Partes_1$ to
define our probabilistic multi-transition systems (pmts).

We can easily translate the 
classical definition of probabilistic bisimulation between 
probabilistic transition systems in \cite{LarsenSkou91}, to our own pmts's as
follows.

\begin{definition} A probabilistic bisimulation
on a coalgebra $p:X\rightarrow\M_{1}([0,1]\times A\times X)$ is an
equivalence relation $\equiv_p$ on $X$ such that, whenever
$x_1\equiv_p x_2$, taking $p(x_i)=\sum
t_{j}^{j}\cdot(p_{j}^{i},a_{j}^{i},x_{j}^{i})$, we also have
$\sum\{t_{j}^{1}\cdot p_{j}^{1}\mid a_{j}^{1}=a,\,x_{j}^{1}\in
E\}$ $=$ $\sum\{t_{j}^{2}\cdot p_{j}^{2}\mid
a_{j}^{2}=a,\,x_{j}^{2}\in E\}$, for all $a\in A$ and any
equivalence class $E$ in $X$/$\equiv_p$.
\end{definition}

In \cite{VinkRutten99} it is proved
that probabilistic bisimilarity
defined by probabilistic bisimulations coincides with categorical
$\D$-bisimilarity. By applying the functor $\D_M$ we can transform our pmts's
into their presentation as Larsen and Skou's pts's. Then it is trivial to check
that the corresponding notions of probabilistc bisimulation coincide, and
therefore they also coincide with categorical $\D$-bisimilarity.

However, that is clearly not the case if we
consider plain categorical $\M_{1}([0,1]\times A\times
\cdot)$-bisimulations. This is so because when we consider the
functor $\M_{1}([0,1]\times A\times \cdot)$, probabilistic
transitions are considered as plain transitions labelled with
pairs over $[0,1]\times A$, whose first component has no special
meaning. As a result, we have, for instance, no bisimulation
relating $x$ and $y$ if we consider $X=\{x\}$, $Y=\{y\}$,
$p_a:X\rightarrow\M_1([0,1]\times A\times X)$ with
$p_a(x)=1\cdot(1,a,x)$ and $p_b:Y\rightarrow\M_1([0,1]\times
A\times Y)$ with $p_b(y)=2\cdot(\frac{1}{2},a,y)$.

All these facts prove that our probabilistic multi-transition systems are too
concrete a representation of probabilistic
distributions, which is formally captured by the fact that the
components of the natural transformation $\D_M$ are not injective. As a
consequence, by using them we do not have a pure coalgebraic characterization
of probabilistic bisimulations. By contrast, the original definition
of pts's stands apart from the operational way, mixing
different
transitions into a single distribution. Besides it has to consider
the quotient set $\quo{X}{\equiv_p}$ when defining probabilistic bisimulations.
Our goal will be to obtain a characterization of the notion of probabilistic
bisimilarity in terms of our pmts's, and this will be done using the
notion of categorical
simulation, as we will see in Section \ref{sec4}. 
Next, we present a collection of general interesting results.
First we will see that bisimulations are preserved by natural transformations.

\begin{theorem}[\cite{Rutten00}]\label{t-bis-pre} If $R\subseteq X\times Y$ is a
bisimulation relating $a:X\lra FX$ and $b:Y\lra FY$, then $R$ is also a 
bisimulation relating $a':X\lra GX$, given by $a'=\alpha_{X}\circ a$, and 
$b':Y\lra GY$, given by $b'=\alpha_{Y}\circ b$.
\end{theorem}

\begin{corollary}
For $a$ and $a' = \alpha_X \mathop{\circ} a$, bisimulation equivalence in 
$a$ is included in bisimulation equivalence in $a'$, that is, 
$x_1\equiv_a x_2$ implies $x_1\equiv_{a'} x_2$.
\end{corollary}

A general converse result cannot be expected because in general
there is no canonical way to transform $G$ into $F$.
Since the main objective in this paper is to relate
$\mult$-bisimulations with $\ps$ and $\disc$-bisimulations, we looked 
for particular properties of the natural
transformations relating these functors which could help us to get
the desired general results covering in particular these two
cases. This is how we have obtained the concept of quotient
functors that we develop in the following.

\begin{definition}
Let $F$ be an endofunctor on $\Sets$ and $\equiv$ a functorial equivalence 
relation $\equiv_{X}\subseteq FX\times FX$. 
We define the \emph{quotient functor} $\quo{F}{\equiv}$ by
$(\quo{F}{\equiv})(X)=\quo{FX}{\equiv_X}$, and for any
$f:X\lra Y$, $u\in FX$, and $\overline{u}$ its equivalence class,
$(\quo{F}{\equiv})(f)(\overline{u})=\overline{F(f)(u)}$, 
that is well defined since $\equiv$ is functorial.
\end{definition}

\begin{definition}
We say that a functor $G$ is the quotient of $F$ under a
functorial equivalence relation $\equiv$, whenever
$\quo{F}{\equiv}$ and $G$ are isomorphic, which means that there
is a pair of natural transformations
$\alpha:\quo{F}{\equiv}\Rightarrow G$ and $\beta:G\Rightarrow\quo{F}{\equiv}$ 
such that $\beta\circ\alpha = Id_{\quo{F}{\equiv}}$ and $\alpha\circ\beta=Id_G$.
\end{definition}

Given a natural transformation $\alpha:F\Rightarrow G$, we write
$\equiv^\alpha$ for the family of equivalence relations
$\equiv^\alpha_X \subseteq FX \times FX$ defined by the kernel of $\alpha$:
\[
u_1 \equiv_X^\alpha u_2 \iff \alpha_{X}(u_1)=\alpha_X(u_2) \,.
\]

\begin{proposition}
For every natural transformation $\alpha:F\Rightarrow G$, 
$\equiv^\alpha$ is functorial.
\end{proposition}
\begin{proof}
We need to show that, for any $f:X\lra Y$, whenever $u_1 \equiv^\alpha_X u_2$,
that is, $\alpha_X(u_1)=\alpha_X(u_2)$, we also have $Ff(u_1)\equiv_Y^\alpha
Ff(u_2)$, that is $\alpha_Y(F(f)(u_1))=\alpha_Y(F(f)(u_2))$; this follows
because $\alpha_Y\circ F(f)=G(f)\circ\alpha_X$.\qed
\end{proof}

If every component $\alpha_X$ of a natural transformation is surjective,
$\alpha$ is said to be epi.

\begin{proposition}
Whenever $\alpha$ is epi, $G$ is the quotient of $F$ under
$\equiv^\alpha$, just considering the inverse natural
transformation $\alpha^{-1}:G\Rightarrow F/{\equiv}$ given by
$\alpha_{X}^{-1}:G(X)\lra(\quo{F}{\equiv^\alpha})(X)$ with
$\alpha_{X}^{-1}(v)=\overline{u}$ where $\alpha_X(u)=v$.
\end{proposition}

\begin{corollary}
$\ps$ is the quotient of $\mult$ under the kernel of the natural
transformation $\{\cdot\}:\mult\Rightarrow\ps$.
\end{corollary}

\begin{corollary}
$\disc$ is the quotient of $\mult_1([0,1]\times \cdot)$ under the kernel
of the natural transformation 
$\disc_{\M}:\mult_1([0,1]\times\cdot)\Rightarrow\disc$.
\end{corollary}

\section{$\equiv^\alpha$-simulations through quotients of
bisimulations}\label{sec4}

Let us start by studying the relationships between coalgebras corresponding 
to functors related by an epi natural transformation.

\begin{definition}
Let $\alpha:F\Rightarrow G$ be a natural transformation and
$a:X\lra FX$ an $F$-coalgebra. 
We define the $\alpha$-image of $a$ as the coalgebra $a_{\alpha}:X\lra GX$ 
given by $a_\alpha=\alpha_X\circ a$.
\end{definition}

\begin{definition}
Given a natural transformation $\alpha:F\Rightarrow G$ and
a $G$-coalgebra $b:X\lra GX$, we say that $a:X\lra FX$ is a concrete 
$F$-representation of $b$ iff $b=\alpha_X\circ a$.
\end{definition}

The following result follows immediately from the previous definitions.

\begin{proposition}\label{p-Frep}
If $\alpha$ is epi then every $G$-coalgebra has an $F$-representation.
\end{proposition}

Next we relate $G$-bisimulations with $\equiv^\alpha$-simulations:

\begin{theorem}\label{t-epi}
Let $\alpha:F\Rightarrow G$ be an epi natural transformation and
$b_1:X_1 \lra GX_1$, $b_2:X_2\lra GX_2$ two $G$-coalgebras, with 
concrete $F$-representations $a_1:X_1\lra FX_1$ and $a_2:X_2\lra GX_2$.
Then, the $G$-bisimulations relating $b_1$ and $b_2$ are precisely the
$\equiv^\alpha$-simulations relating $a_1$ and $a_2$.
\end{theorem}
\begin{proof}
\footnote{It is not difficult to present this proof as a commutative
diagram. Then one has to check that all the ``small squares'' in the diagram
are indeed commutative, in order to be able to conclude commutativity of the
full diagram. This is what we have carefully done in our proof above.}

Let us show that, for every relation $R\subseteq X_1\times X_2$,
\[
\rel{F}_{\equiv^\alpha}(R) =
\{(u,v)\in FX_1\times FX_2\mid (\alpha_{X_1}(u),\alpha_{X_2}(v))\in\rel{G}(R)\} \,.
\]
We have, unfolding the definition of $\rel{F}_{\equiv^\alpha}(R)$ and using the
fact that $\alpha$ is a natural transformation:
\[
\begin{array}{rcl}
\rel{F}_{\equiv^\alpha}(R) &= 
&\{ (u,v) \in FX_1\times FX_2 \mid
   \exists w\in FR.\, u\equiv^\alpha Fr_1(w) \land Fr_2(w)\equiv^\alpha v \} \\
&=&\{ (u,v) \in FX_1\times FX_2 \mid
   \exists w\in FR.\, \alpha_{X_1}(u) = \alpha_{X_1}(Fr_1(w)) \land{}\\
&&\phantom{\{ (u,v) \in FX_1\times FX_2 \mid \exists w\in FR.\,}
                      \alpha_{X_2}(v) = \alpha_{X_2}(Fr_2(w)) \} \\
&=&\{ (u,v) \in FX_1\times FX_2 \mid
   \exists w\in FR.\, \alpha_{X_1}(u) = Gr_1(\alpha_R(w)) \land{}\\
&&\phantom{\{ (u,v) \in FX_1\times FX_2 \mid \exists w\in FR.\,}
                      \alpha_{X_2}(v) = Gr_2(\alpha_R(w)) \}\,.
\end{array}
\]
On the other hand,
\[
\rel{G}(R) = \{ (x,y)\in GX_1\times GX_2\mid 
                \exists z\in GR.\, Gr_1(z)= x \land Gr_2(z)=y \} \,.
\]
Now, if $(u,v)\in \rel{F}_{\equiv^\alpha}(R)$, by taking $\alpha_R(w)$ as the value
of $z\in GR$ we have that $(\alpha_{X_1}(u), \alpha_{X_2}(v)) \in \rel{G}(R)$.
Conversely, if $(\alpha_{X_1}(u), \alpha_{X_2}(v)) \in \rel{G}(R)$ is witnessed
by $z$, let $w\in FR$ be such that $\alpha_R(w) = z$, which must exists because
$\alpha$ is epi; it follows that $(u,v)\in \rel{F}_{\equiv^\alpha}(R)$.

Then, $(b_1(x), b_2(y))\in \rel{G}(R)$ if and only if
$(a_1(x), a_2(x)) \in \rel{F}_{\equiv^\alpha}(R)$, from where it follows that $R$
is a $G$-bisimulation if and only if it is a $\equiv^\alpha$-simulation.
\qed
\end{proof}

\begin{corollary}
(i) Bisimulations between labelled transition systems are just 
$\equiv^{\{\cdot\}}$-simulations between multi-transition systems.
(ii) Bisimulations between probabilistic systems are just
$\equiv^{\disc_M}$-simulations between (an appropriate class of) 
multi-transition systems.
\end{corollary}

Let us illustrate this result by means of some simple examples using
the natural transformation $\{\cdot\}:\mult\rightarrow\ps$. 
If we consider the ordinary transition systems
$s_X:\{x,x'\}\lra\ps(\{x,x'\})$, with $s_X(x)=\{x'\}$, $s_X(x')=\emptyset$, and
$s_Y:\{y,y'_1,y'_2\}\lra\ps(\{y,y'_1,y'_2\})$ with
$s_Y(y)=\{y'_1,y'_2\}$, $s_Y(y'_1)=\emptyset$, and
$s_Y(y'_2)=\emptyset$, we have a simple $\ps$-bisimulation
relating the initial states $x$ and $y$, given by
$R=\{(x,y),(x',y'_1),(x',y'_2)\}$.

Denoting by $s_X^1$ and $s_Y^1$ the canonical $\mult$-representations of $s_X$ 
and $s_Y$, obtained by the embedding of sets into multisets,
it is obvious that there is no $\mult$-bisimulation relating $x$ and $y$. 
But if we consider $s_X^2(x)=\{2\cdot x'\}$, $s_X^2(x')=\emptyset$, we have 
now an $\mult$-bisimulation between the
multi-transition systems $s_X^2$ and $s_Y^1$ relating $x$ and $y$. 
And, by Theorem~\ref{t-epi}, we have that $s_X^1$ is also 
$\equiv^{\{\cdot\}}$-simulated by $s_Y^1$, since $\{s_X^1\}_\M=\{s_X^2\}_\M=s_X$
and $s_X$ and $s_Y$ are $\ps$-bisimilar. Obviously, the
same happens for any $\{\cdot\}$-representation of $s_X$, $s_X^k$
with $s_X^k=\{k\cdot x'\}$ and $s_X^k(x')=\emptyset$.

In the example above we got the $\equiv^{\{\cdot\}}$-simulation by proving 
that there are $\mathcal{M}$-representations of the considered coalgebras 
for which the given relation is also an $\mathcal{M}$-bisimulation. 
However, this is not necessary as the following counterexample shows:

Let us consider $t_X:\{x\}\lra\ps(\{x\})$ with
$t_X(x)=\{x\}$ and $Y=\{\beta\mid\beta \in \mathbb{N^*}, \beta_i\leq i\}$ with
$t_Y(\beta)=\{\beta\circ\langle j\rangle\mid \beta\circ\langle j\rangle\in Y\}$.
It is clear that $R=\{(x,\beta)\mid \beta\in Y\}$ is (the only) 
$\ps$-bisimulation relating $x$ and $\epsilon$, the initial states of $t_X$ 
and $t_Y$. 
However, in this case there exists no $\mult$-bisimulation
relating two $\mult$-representations of $t_X$ and $t_Y$, because
$|t_Y(\beta)|=|\beta|+1$ and therefore we would need a representation 
$t^k_X$ with $t^k_X(x)=\{k\cdot x\}$ such that $k\geq l$ for all 
$l\in\mathbb{N}$, which is not possible because the definition of 
multiset does not allow the infinite repetition of any of its members.
Instead, Theorem~\ref{t-epi} shows that any two $\mult$-representations of 
$t_X$ and $t_Y$ are $\equiv^{\{\cdot\}}$-similar.

The reason why we had an $\mult$-bisimulation relating the appropriate 
$\M$-representations of the compared $\ps$-coalgebras in our first
example, was because we were under the
hypothesis of the following proposition.

\begin{proposition}\label{p-Ref}
Let $\alpha: F\Rightarrow G$ be an epi natural transformation.
Whenever a $G$-bisimulation $R$ relating $b_1:X\lra GX$ and $b_2:Y\lra GY$ is 
\emph{near injective}, which means
that $|\{ b_2(y) \mid (x,y) \in  R \}|\leq 1$ for all $x\in X$ and
$|\{b_1(x)\mid (x,y)\in R\}|\leq 1$ for all $y\in Y$, there
exist some $F$-representations of $b_1$ and $b_2$, $a_1:X\lra FX$ and 
$a_2:Y\lra FY$, respectively, such that $R$ is also a
bisimulation relating $a_1$ and $a_2$.
\end{proposition}
\begin{proof}
By Theorem~\ref{t-epi}, $R$ is also a $\equiv^\alpha$-simulation for any
pair of $F$-representations of $b_1$ and $b_2$; let $a_1$, $a_2$ be any
such pair.
Then, for all $(x,y)\in R$ we have 
$(a_1(x),a_2(y)) \in 
(\equiv^\alpha \mathop{\circ}\rel{F}\mathop{\circ}\equiv^\alpha)(R)$, 
and hence there exist $a'_1(x,y)\in FX$, $a'_2(x,y)\in FY$ such that 
\[
a_1(x) \equiv^\alpha a'_1(x,y),\, 
a'_2(x,y)\equiv^\alpha a_2(y)\,\textrm{ and }\,
(a'_1(x,y),a'_2(x,y))\in\rel{F}(R) \,.
\] 

We now define an equivalence relation $\equiv$ on $R$ by considering the
transitive closure of:
\begin{itemize}
\item $(x,y_1)\equiv (x,y_2)$ for all $(x,y_1)$, $(x,y_2)\in R$.  
\item $(x_1,y)\equiv (x_2,y)$ for all $(x_1,y)$, $(x_2,y)\in R$.
\end{itemize}
Since $R$ is near injective, it follows that if $(x_1,y_1)\equiv (x_2,y_2)$ 
then $b_1(x_1)=b_1(x_2)$ and $b_2(y_1)=b_2(y_2)$, and thus
$a'_1(x_1,y_1) \equiv^\alpha a'_1(x_2,y_2)$ and 
$a'_2(x_1,y_1) \equiv^\alpha a'_2(x_2,y_2)$.

We consider $R/{\equiv}$ and for each equivalence class of the 
quotient set we choose a canonical representative $\overline{(x,y)}$. 
Obviously we have that $(x,y_1),(x,y_2)\in R$ implies
$\overline{(x,y_1)}=\overline{(x,y_2)}$ and that
$(x_1,y),(x_2,y)\in R$ implies $\overline{(x_1,y)}=\overline{(x_2,y)}$.

Let us now define two coalgebras $a'_1: X\lra FX$ and $a'_2 : Y \lra FY$ 
as follows:
\begin{itemize}
\item If there exists some $y$ such that $(x,y)\in R$ we take 
 $a'_1(x)=a'_1\overline{(x,y)}$ for any such $y$; otherwise, we
 define $a'_1(x)$ as $a_1(x)$. 
\item If there exists some $x$ such that $(x,y)\in R$ we take 
 $a'_2(y)=a'_2\overline{(x,y)}$ for any such $x$; otherwise,
 $a'_2(y)$ is $a_2(y)$.
\end{itemize}
With the above definitions, 
\[
a'_1(x) = a'_1\overline{(x,y)} \equiv^\alpha a'_1(x,y) \equiv^\alpha a_1(x) \,,
\]
and similarly $a'_2(y) \equiv^\alpha a_2(y)$, so that $a'_1$, $a'_2$ are
$F$-representations of $b_1$ and $b_2$.
Besides,
\[
\textrm{if $(x,y)\in R$ then $(a'_1(x),a'_2(y)) \in\rel{F}(R)$} 
\]
and $R$ is an $F$-bisimulation relating them. 
\qed
\end{proof}

Let us conclude this illustration of our main theorem, by explaining 
why we needed an infinite coalgebra to get a counterexample of the result
between
bisimulations relating $G$-coalgebras and those relating their
$F$-representations. As a matter of fact, in the case of the
multiset and the powerset functors we could prove the result in
Prop.~\ref{p-Ref} not only for near injective bisimulations but
for any relation where no element is related with infinitely many
others. However, we will not prove this fact here since it does not seem 
to generalize to arbitrary natural transformations relating two
functors.

Next we present more examples for the natural transformation
$\disc_{M}:\mult_1([0,1]\times X)\Rightarrow \disc(X)$. If we consider the
two probabilistic transition systems $s_X$ and $s_Y$ given by
their (multi)sets of probabilistic transitions:
$s_X=\{(\frac{1}{2},x,x'_1),(\frac{1}{2},x,x'_2)\}$,
$s_Y=\{(\frac{1}{3},y,y'_1),(\frac{1}{3},y,y'_2),(\frac{1}{3},y,y'_3)\}$,
where each $3$-tuple $(p,x,x')$ represents the probabilistic
transition $x\stackrel{p}{\rightarrow}x'$, we have the following
$\disc$-bisimulation relating the initial states $x$ and $y$:
$R=\{(x,y)\} \cup \{(x'_i,y'_j) \mid i\in 1,2, j\in 1,\ldots 3\}$. 
It is easy to see that considering the two
$\mult_1$-representations $s_X^3=\{3\cdot (\frac{1}{6},x,x'_1),
3\cdot (\frac{1}{6},x,x'_2)\}$ and $s_Y^2=\{2\cdot
(\frac{1}{6},y,y'_1), 2\cdot (\frac{1}{6},y,y'_2),2\cdot
(\frac{1}{6},y,y'_3)\}$, $R$ gives also an $\mult_1$-bisimulation
between them, using the facts that
$(x'_1,y'_1)\in R$, $(x'_2,y'_2)\in R$ and $(x'_1,y'_3)\in R$,
$(x'_2,y'_3)\in R$. From this result we immediately conclude that
any two $\mult_1$-representations of $s_X$ and $s_Y$ are
$\equiv^{\D_\M}$-similar.


\section{Natural transformations and simulations}

In this section we will see that all our results about
bisimulations in the previous sections can be extended to
categorical simulations defined by means of an order on the
corresponding functors. Therefore, our first result concerns the
preservation of functorial orders by means of natural
transformations.

\begin{definition}
Given a natural transformation $\alpha:F\Rightarrow G$ and
$\sqsubseteq_G$ an order on $G$, we define the induced order 
$\sqsubseteq_{G}^{\alpha-}$ on $F$ by
\[
x \<_G^{\alpha-} x' \iff \alpha_X(x)\<_G \alpha_X(x')
\]
\end{definition}


It is immediate that $\<_G^{\alpha-}$ is indeed an order on $F$;
given $f : X \lra Y$ and $x, x' \in X$:
\[
\begin{array}{rcl}
x\<_G^{\alpha-} x' &\iff &\alpha_X(x) \<_G \alpha_X(x') \\
&\Longrightarrow &Gf(\alpha_X(x)) \<_G Gf(\alpha_X(x')) \\
&\iff &\alpha_Y(Ff(x)) \<_G \alpha_Y(Ff(x')) \\
&\iff &Ff(x) \<_G^{\alpha-} Ff(x') \,,
\end{array}
\]
where the implication follows because $\<_G$ is functorial.

An interesting example corresponds to the equality relation as an
order on $G$.

\begin{proposition}
$=_G^{\alpha-}$ is $\equiv^\alpha_G$.
\end{proposition}
\begin{proof}
The definition of $\equiv^\alpha_G$ is just the
particular case of our definition of $\<_G^{\alpha-}$ for the
equality relation on $G$ as an order on it. 
\qed
\end{proof}

Orders on $F$ can be also translated to $G$ through a natural
transformation $\alpha: F\Rightarrow G$.

\begin{definition}
Given a natural transformation $\alpha: F\Rightarrow G$ and $\<_F$
an order on $F$, we define the projected order $\<_F^\alpha$ on $G$ as the
transitive closure of the relation $x\<_F^{\alpha} x'$, which holds if and 
only if
\[
\textrm{there exist $x_1$, $x'_1$ such that $x=\alpha_X(x_1)$, 
$x'=\alpha_X(x'_1)$ and $x_1\<_F x'_1$, or $x=x'$.}
\]
\end{definition}

We need to add the last condition in the definition above in order
to cover the case in which $\alpha$ is not an epi. Obviously, we can
remove it whenever $\alpha$ is indeed an epi, and in the following
we will see that we only need that condition in order to guarantee
reflexivity of $\<_F^\alpha$ in the whole of $GX$, because all of
our results concerning this order will be based on its restriction
to the images of the components of the natural transformation
$\alpha_X$.

Again, it is easy to prove that $\<_F^\alpha$ is indeed an order on $G$.
By definition, it is reflexive and transitive.
It is also functorial: given 
$f: X\lra Y$ and $x\<_F^\alpha x'$, with $x=\alpha_X(x_1)$ and
$x'=\alpha(x'_1)$ such that $x_1 \<_F x'_1$, 
we need to show $Gf(x) \<_F^\alpha Gf(x')$.
Since $Gf(x) = Gf(\alpha(x_1)) = \alpha(Ff(x_1))$, 
$Gf(x') = Gf(\alpha(x'_1)) = \alpha(Ff(x'_1))$, and 
$Ff(x_1) \<_F Ff(x'_1)$, the result follows by the definition of $\<_F^\alpha$.

\begin{theorem}[Simulations are preserved by natural transformations] 
If $R\subseteq X\times Y$ is a $\<_F$-simulation relating $a : X\lra FX$ and
$b : Y\lra FY$, and $\alpha:F\Rightarrow G$ is a natural transformation, then 
$R$ is also a $\<_F^\alpha$-simulation relating $a'=\alpha_X \circ a$ and
$b'=\alpha_Y \circ b$.
\end{theorem}
\begin{proof}
Let $(x,y) \in R$: we need to show that $(a'(x), b'(y))\in 
\rel{G}_{\<_F^\alpha}(R)$.
Since $R$ is a $\<_F$-simulation, $(a(x), b(x)) \in \rel{F}_{\<_F}(R)$.
This means that there exists $w\in FR$ such that
$a(x) \<_F Fr_1(w)$ and $Fr_2(w) \<_F b(x)$, and hence that there exists 
$z = \alpha_R(w)\in GR$ such that $a'(x) \<_F^\alpha \alpha_X(Fr_1(w)) = Gr_1(z)$ 
and $Gr_2(z)= \alpha_Y(Fr_2(w)) \<_F^\alpha b'(x)$; therefore,
$(a'(x), b'(x)) \in \rel{G}_{\<_F^\alpha}(R)$. 
\qed
\end{proof}

As said before, bisimulations are just the particular case of
simulations  corresponding to the equality relation.
Obviously we have that $=_F^\alpha$ is $=_G$ and therefore 
Theorem~\ref{t-bis-pre} about the preservation of bisimulations by natural 
transformations is a particular case of our new preservation theorem
covering arbitrary $\<_F$-simulations.

Let us now also extend our result about the reflection of
bisimulations to arbitrary $\<_G$-simulations.

\begin{theorem}
Let $\alpha:F\Rightarrow G$ be an epi natural transformation, $\<_G$
an order on $G$ and $b_1 :X_1 \lra GX_1$, $b_2 :X_2\lra GX_2$ two
coalgebras, with $a_1 : X\lra FX$, $a_2 : Y\lra FY$
arbitrary concrete $F$-representations. 
Then, the $\<_G$-simulations relating $b_1$ and $b_2$ are precisely the 
$\<_G^{\alpha-}$-simulations relating $a_1$ and $a_2$.
\end{theorem}
\begin{proof}
Analogously to Theorem~\ref{t-epi}, the result follows from showing that, 
for every relation $R\subseteq X_1\times X_2$,
\[
\rel{F}_{\<_G^{\alpha-}}(R) =
\{(u,v)\in FX_1\times FX_2\mid 
       (\alpha_{X_1}(u),\alpha_{X_2}(v))\in\rel{G}_{\<^{\alpha}_G}(R)\} \,.
\]
Unfolding the definition of $\rel{F}_{\<_G^{\alpha-}}(R)$ and using the
fact that $\alpha$ is a natural transformation:
\[
\begin{array}{rcl}
\rel{F}_{\<_G^{\alpha-}}(R) &=
&\{ (u,v) \in FX_1\times FX_2 \mid
   \exists w\in FR.\, u\<_G^{\alpha-} Fr_1(w) \land{} \\
&&\phantom{\{ (u,v) \in FX_1\times FX_2 \mid \exists w\in FR.\,} 
                      Fr_2(w)\<_G^{\alpha-} v \} \\
&=&\{ (u,v) \in FX_1\times FX_2 \mid
   \exists w\in FR.\, \alpha_{X_1}(u) \<_G \alpha_{X_1}(Fr_1(w)) \land{}\\
&&\phantom{\{ (u,v) \in FX_1\times FX_2 \mid \exists w\in FR.\,}
                      \alpha_{X_2}(Fr_2(w)) \<_G \alpha_{X_2}(v) \} \\
&=&\{ (u,v) \in FX_1\times FX_2 \mid
   \exists w\in FR.\, \alpha_{X_1}(u) \<_G Gr_1(\alpha_R(w)) \land{}\\
&&\phantom{\{ (u,v) \in FX_1\times FX_2 \mid \exists w\in FR.\,}
                      Gr_2(\alpha_R(w)) \<_G \alpha_{X_2}(v) \}\,.
\end{array}
\]
On the other hand,
\[
\rel{G}_{\<_G}(R) = \{ (x,y)\in GX_1\times GX_2\mid 
                \exists z\in GR.\, x \<_G Gr_1(z) \land Gr_2(z)\<_G y \} \,.
\]
Now, if $(u,v)\in \rel{F}_{\<_G^{\alpha-}}(R)$, by taking $\alpha_R(w)$ 
as the value of $z\in GR$ we have that 
$(\alpha_{X_1}(u), \alpha_{X_2}(v)) \in \rel{G}_{\<_G}(R)$.
Conversely, if $(\alpha_{X_1}(u), \alpha_{X_2}(v)) \in \rel{G}_{\<_G}(R)$ 
is witnessed by $z$, let $w\in FR$ be such that 
$\alpha_R(w) = z$, which must exist because
$\alpha$ is epi; it follows that $(u,v)\in \rel{F}_{\<_G^{\alpha-}}(R)$.
\qed
\end{proof}

Once again this result on $\<$-simulations generalizes that on
bisimulations in Theorem \ref{t-epi}, because if we take the
equality relation $=_G$ as an order on $G$ and $\alpha$ as an epi
natural transformation, we have that $=_G^{\alpha-}$ is $\equiv_\alpha$.

\section{Combining non-determinism and probabilistic choices}

Probabilistic choice appears as a quantitative counterpart of non-deterministic
choice. However, it has been also argued that the motivations supporting the
use of these two constructions are different, so that it is also interesting to
be able to manage both together. The literature on the subject is full of
proposals in this direction \cite{SL95,Mis00,TKP05}, but it has been
proved in \cite{VarWins06} that there is no distributive law of the
probabilistic monad $V$ over the powerset monad $P$. As a consequence, if we
want to combine the two categorical theories to obtain a common framework, we
have to sacrifice some of the properties of one of those monads. Varacca and
Winskel have followed this idea by relaxing the definition of the monad $V$,
removing the axiom $A\oplus_p A=A$, so that they are aware of the probabilistic
choices taken along a computation, even if they are superfluous.

We have not yet studied that general case, whose solution in \cite{VarWins06} is
technically fine, but could be considered intuitively not too satisfactory,
since one would like to maintain the idempotent law $A\oplus_p
A=A$, even if this means that only some practical cases can be considered.

As a first step in this direction we will present here the simple case of
alternating probabilistic systems, which in our multi-transition system
framework can be defined as follows:

\begin{definition}
Alternating multi-transition systems are defined as $(\M(A\times
\cdot)\cup\M_1([0,1]\times A\times \cdot))$-coalgebras: any state of a system
represents either a non-deterministic choice or a probabilistic choice; however,
probabilistic and non-deterministic choices cannot be mixed together.
\end{definition}

By combining the two natural transformations $\{\cdot\}$ and $\D_M$ we 
obtain the natural transformation $\D_{M}^a$, that captures the behaviour of
alternating transition systems.

\begin{definition}
We call alternating probabilistic systems to the
$\Partes(X)\cup\D(A\times X)$-coalgebras. By combining the classical definition
of bisimulation and that of probabilistic bisimulations we obtain the natural
definition of probabilistic bisimulation for alternating probabilistic
systems.
We define $\D_{M_X}^a:\M(A\times
X)\cup\M_1([0,1]\times A\times X)\Longrightarrow \Partes(X)\cup\D(A\times X)$
as $\D_{M_X}^a(M)=\{\cdot\}(M)$, $\D_{M_X}^a(M_1)=\D_M(M_1)$, where
$M\in\M(A\times X)$, $M_1\in\M_1([0,1]\times A\times X)$.
\end{definition}

Then we can consider the induced functorial equivalence $\equiv^{\D_M^a}$ which
roughly corresponds to the application of $\equiv^{\{\cdot\}}$ in the
non-deterministic states, and the application of $\equiv^{\D_M}$ in the
probabilistic states. As a consequence of Theorem \ref{t-epi} we obtain the
following corollary.

\begin{corollary}
Bisimulations between alternating probabilistic systems are just
$\equiv^{\D_M^a}$-simulations between alternating multi-transition
systems.
\end{corollary}

\section{Conclusion}

In this paper we have shown that multitransition systems are a
common framework wherein bisimulation of ordinary and probabilistic transition 
systems almost collapse into the same concept of multiset (bi)simulation. 
Indeed, the definition of bisimulation for the multiset functor is extremely 
simple, which supports the idea that multisets are the natural framework 
in which to justify the use of bisimulation as the canonical notion 
of equivalence between (states of) systems.

These results have been obtained by exploiting the fact that natural
transformations between two functors relate in a nice way bisimulations over
their corresponding coalgebras. We have illustrated these general results by
means of the natural transformations that connect the powerset and the
probabilistic distributions functors with the multiset functor. 

The categorical notion of simulation proposed by Hughes and Jacobs has played 
a very important role in our work; this fact, in our opinion, is far from 
being casual. 
In particular, categorical simulations based on equivalence relations 
always define equivalence relations weaker than bisimulation equivalence.
Besides, as illustrated by their use in this paper, they can be used to 
relate the bisimulation equivalence corresponding to functors connected 
by a natural transformation.


Related to our work is \cite{BartelsEtAl04}, where probabilistic bisimulations 
are studied in connection with natural transformations and other categorical 
notions.
Even though some connections can be found, there are very important differences;
in particular they do not consider categorical simulations nor use the 
multiset functor as a general framework in which to study both ordinary and 
probabilistic bisimulations. 
We can also mention \cite{Varacca02}, where the functor $\mathcal{D}$ is
replaced with a functor of indexed valuations so that it can be combined
with the powerset functor. 



A direction for further study that we intend to explore concerns other 
classes of bisimulations, like the forward-backward ones estudied in 
\cite{Hasuo06}. Besides we will study more general combinations of
non-deterministic and probabilistic choices, comparing in detail our
approach with the use of indexed valuations in \cite{Varacca02,VarWins06} to
combine the monads defining the corresponding functors. 

We are confident we will
be able to study them in a common setting by 
generalizing and adapting all the appropriate notions on categorical 
simulations.



\end{document}